%% file: fdsr.tex
\documentclass[confsoc,a4paper,12pt]{IEEEtran}

\usepackage[T1]{fontenc}
 
\usepackage[latin1]{inputenc}
\usepackage{graphicx}
\usepackage{amsfonts}
\usepackage{amsmath}
\usepackage{mathrsfs,theorem}
\usepackage{rotating}
\usepackage{algorithm, algorithmic}
\floatname{algorithm}{Algorithm}
\usepackage[usenames,dvipsnames]{color}
\usepackage{multirow}
\usepackage{dcolumn}

\usepackage{graphicx}
\usepackage{mathrsfs}
\usepackage{xspace}
\usepackage{algorithmic}
\usepackage{algorithm}
\usepackage{dcolumn}
\usepackage{url}
\usepackage{algorithm, algorithmic}
\floatname{algorithm}{Algorithm}

\def\..{\,\mathpunct{\ldotp\ldotp}} % Middle stuff for intervals. Usage: \..

\newtheorem{theorem}{Theorem}

\newcounter{noqed}
\newcommand{\qed}{ \ifmmode\mbox{ }\fi\rule[-.05em]{.3em}{.7em}\setcounter{noqed}{0}}
\newenvironment{proof}[1][{}]{\noindent{\bf Proof#1. }\setcounter{noqed}{1}}{\ifnum\value{noqed}=1\qed\fi\par\medskip}

\title{Four Degrees of Separation, Really}

\author{\IEEEauthorblockN{Paolo Boldi \qquad Sebastiano
Vigna}\\\IEEEauthorblockA{Dipartimento di Informatica\\Universit\`a degli Studi
di Milano\\Italy}\thanks{Partially supported by a Yahoo!~faculty grant and by 
by the EU-FET grant NADINE (GA 288956).}
}

\begin{document}
\bibliographystyle{IEEEtran}
\maketitle

\begin{abstract}
We recently measured the average distance of users in the Facebook graph,
spurring comments in the scientific community as well as in the general
press~\cite{BBRFDS}.
A number of interesting criticisms have been made about the meaningfulness,
methods and consequences of the experiment we performed. In this paper
we want to discuss some methodological aspects that we deem important to underline in the form of
answers to the questions we have read in newspapers, magazines, blogs, or heard
from colleagues.
We indulge in some reflections on the
actual meaning of ``average distance'' and make a number of side observations
showing that, yes, $3.74$ ``degrees of separation'' are really
few.
\end{abstract}

\section*{Four degrees of separation}

In 2011, together with Marco Rosa, we developed a new tool for studying the
distance distribution of very large (unweighted) graphs, called
HyperANF~\cite{BRVH}:
this algorithm built on powerful graph compression techniques~\cite{BoVWFI} and on the idea of diffusive computation pioneered
in~\cite{PGFANF}. The new tool made it possible to accurately study the
distance distribution of graphs orders of magnitude larger than it was previously possible.
The work on HyperANF was presented at the 20th World-Wide Web Conference, in
Hyderabad (India), and Lars Backstrom happened to listen to the talk; he
was intrigued by the possibility of experimenting our software on the
Facebook graph and suggested a collaboration. 

Experiments were performed in the summer of 2011, resulting in the first
world-scale social-network graph-distance computations, using the entire
Facebook network of active users (721 million users, 69 billion friendship
links). The average distance (i.e., shortest-path length) observed was $4.74$, corresponding to $3.74$ intermediaries (or
``degrees of separation'', in Milgram's parlance).
These and other findings were finally
presented in~\cite{BBRFDS} and made public by Facebook through its technical
blog on November 19, 2011. Immediately after the announcement, the news  
appeared in the general press, starting from the New York
Times~\cite{MaSSYMFD}\footnote{Incidentally, with an off-by-one error, as
$4.74$ is the average \emph{distance}, whereas the average number of degrees
of separation is $3.74$ (see~\cite{BBRFDS}).} and soon spreading worldwide in
newspapers, blogs and forums.

A number of interesting criticisms have been made about the
meaningfulness, methods and consequences of the experiment we performed. In this paper
we want to discuss some methodological aspects that we deem important. We shall
consider such issues in an answer-to-question style, with the double aim of
replying to doubts and attacks and of stimulating new discussions and further
interest.

\section{Not all pairs are connected: how can the average distance be
even finite?}
\label{sec:dist}

If by ``average distance'' we mean ``average of the distances between all
pairs'', of course Facebook has an infinite average distance, as we know that
there is a very large connected component containing almost all ($99.9\%$)
nodes, but there are also some (few) unreachable pairs.

This is an interesting comment, as it shows an actual black hole in all the
literature: people studying social problems (starting with the 50s, at least)
had in mind very small groups, possibly groups that would fit one room
(actually, in some cases, just sitting around a table).
Or small communities. The very idea of ``unreachable'' was not part of the
picture.
In the famous paper by Travers and Milgram~\cite{TMESSWP}, the vast majority of 
postcards did not reach the target\footnote{It should be noted, as an aside,
that in Milgram's experiment the interrupted chains do not actually imply
unreachability, a point that will be better discussed later.}.
Nonetheless, the ``six degrees of separation'' idea came from the average distance ($5.4$ to $6.7$, depending on
the group) obtained in the experiment, computed \emph{just on
reachable pairs}.\footnote{Indeed, the authors of one of the first studies of the web as a
whole~\cite{BKMGSW} noted the same problem, and proposed the name \emph{average
connected distance}. We refrain, however, from using the word ``connected'' as
it somehow implies a bidirectional (or, if you prefer, undirected) connection. 
The notion of average distance between all pairs
is useless in a graph in which not all pairs are reachable, as it is necessarily
infinite, so no confusion can arise.}

We discuss here in some detail two possible mathematical solutions to this
problem---not only because they are interesting, but because we
want to urge researchers to take the problem into consideration more seriously,
and to remark to those objecting to the use of reachable pairs that
old results would be really stated differently if unreachable pairs were
correctly taken into account.

An obvious patch is to quote the average distance between reachable pairs, sided
by the percentage of reachable pairs, which should be considered as a sort of
\emph{confidence} on the measure. If the percentace of reachable pairs
is low, the average distance is telling us little. On a completely disconnected
graph, the average distance is 0, but with ``confidence'' $1/n$. On a
perfect match,\footnote{A \emph{perfect match} is an undirected 1-regular
graph, that is, a set of disconnected edges.} the average distance is $1/2$,
but the ``confidence'' is $2/n$ (in both cases, almost zero for large graphs).

Seen in this perspective, Milgram's experiment proposes an average distance of
$6.2$ but provides an incredibly low level of confidence---just
$22\%$,\footnote{Travers and Milgram's paper~\cite{TMESSWP} reports $29\%$, as
this is the percentage of chains that \emph{started and completed} with respect to those that \emph{started}. 
Some of the chains did not start at all, and we are
considering them as incomplete, which explains the slightly slower value we
are reporting.} whereas in our case we can claim confidence $99.9\%$ for our
value ($4.74$).

The problem is that we like to compare results, and comparing two pairs of
numbers can be difficult, if not impossible (see, e.g., the plethora of methods
used to combine somehow precision and recall in information retrieval).

A solution that does not show the latter drawback is to consider \emph{harmonic
means} when working with distances. We recall that the harmonic mean is the reciprocal of the mean of the
reciprocals. It is always smaller than the arithmetic mean, as it tends to give
less relevance to large outliers and more relevance to small values, and it is
used in a number of contexts\footnote{Incidentally, the HyperLogLog
counters~\cite{FFGH} used by HyperANF~\cite{BRVH}, the algorithm with which the
average distance of Facebook was computed, use the harmonic mean to perform stochastic
averaging.}.

The important feature of the harmonic mean is that if we stipulate that
$1/\infty=0$, it can take in $\infty$ as a perfectly valid distance. Its effect
is that of making the mean larger in a hyperbolic fashion. This is why Marchiori
and Latora~\cite{MaLHSW} proposed to consider the harmonic mean of
\emph{all} distances between distinct nodes\footnote{The fact that we do not
consider the distances $d(x,x)$ is essential, as otherwise the harmonic mean
becomes zero.}, which we call \emph{harmonic diameter} following
Fogaras~\cite{FogWSBW} (rather than ``average distance \emph{between reachable pairs}''), as a measure of tightness of a network. For instance, a disconnected graph has average distance zero, but infinite harmonic diameter; and a perfect match has average
distance $1/2$, but harmonic diameter $n-1$.

What happens if we switch from the average distance to the harmonic diameter? On
highly disconnected network, with many missing paths, we get a larger number. On
the LAW web site\footnote{\texttt{http://law.dsi.unimi.it/}} you can find the
basic statistics of several web-graph snapshots, and the harmonic diameter is 
always significantly larger than the
average distance between reachable pairs.

\begin{table*}[ht]
\caption{\label{tab:harmonicno2007} Harmonic diameter of the graphs
from~\cite{BBRFDS}.}
\begin{center}
\begin{tabular}{|c|r|r|r|r|r|}
	\hline
	& \multicolumn{1}{c|}{\texttt{it}} & \multicolumn{1}{c|}{\texttt{se}} & \multicolumn{1}{c|}{\texttt{itse}} & \multicolumn{1}{c|}{\texttt{us}} & \multicolumn{1}{c|}{\texttt{fb}} \\
	\hline
	\input scripts/harmonic.table
	\hline
\end{tabular}
\end{center}
\end{table*}

In the case of Facebook, the harmonic diameter is $4.59$---even smaller than the
average distance. The situation, however, is quite different if we make the same
computation with Milgrams' experiment and assume that incomplete chains
correspond to unreachable pairs:
overall, the harmonic mean is $18.29$, almost four times larger than the average distance. If we restrict to the
Nebraska random group (i.e., we avoid geographical or cultural clues), the harmonic mean is more than five times
larger. By this measure, the improvement described in~\cite{BBRFDS} is even more
impressive.

\begin{table}
\caption{The harmonic mean and the mean of all distances (including $\infty$
for broken chains) for the groups detailed in Travers and Milgram's
paper~\cite{TMESSWP}. Note the significantly lower value of the harmonic mean
for the Boston group.}
\centering
%\begin{tabular}{lD{.}{.}{3}c}
\begin{tabular}{|l|D{.}{.}{3}|c|}
\hline
\multicolumn{1}{|c|}{Group} & \multicolumn{1}{|c|}{Harmonic mean}
&\multicolumn{1}{|c|}{Median distance}\\
\hline
Nebraska random     & 	26.68&$\infty$ \\
Nebraska stockholders&	19.37&$\infty$\\
All Nebraska         &	22.40&$\infty$\\
Boston random        &	12.63&$\infty$\\
All                  &	18.29&$\infty$\\
\hline
\end{tabular}
\end{table}

The problem with the harmonic diameter is that even if it is a clearly and
sensibly defined mathematical feature, it deprives us from the ``degree of
separation'' metaphore. The fact that in 2007 the harmonic diameter of
\texttt{it} was more than $15\,000$ does not mean, of course, that you need to
pass through $15\,000$ friendship links!

Another possibility for taking into account infinite distances is to use the 
\emph{median of all distances} as a
measure of closeness. That is, we list in
increasing order the $n^2$ values of $d(x,y)$, and we take that of index
$\lfloor n^2/2\rfloor$ (numbering from zero). This number is significantly
larger than the average distance if several pairs are unreachable because the $\infty$ values at
the end of the list ``push'' the median to the right. Again, on
the LAW web site you can see that in several web graphs the median of all
distances is significantly larger than the average distance, as it takes into
account the existence of unreachable pairs. It is a good idea to complement the
median with the fraction of pairs within its value: in any case, we know that at
least $50\%$ of the pairs (of \emph{all} pairs, not just the reachable ones) are
within its value, which gives us a concrete handle.

The median of all distances for Facebook is 5 (and $92\%$ of all pairs is
within this distance). So, again, ``four degrees of separation''. Obviously, for
Milgram in all cases the median is $\infty$. So, using this measure we
progressed really a lot.

With the collaboration of Jure Leskovec we were able to 
compute similar measures for Horvitz and Leskovec's Messenger
experiment~\cite{LHPSVLIMN}: the 
average distance, $6.618$, has confidence $71.3\%$; the harmonic 
diameter is $8.935$, whereas the median distance is $7$, covering $78.7\%$ of
all pairs.\footnote{We cannot report statistical metadata such as the standard
error, because we were provided with already-aggregated breadth-first samples
only.} Note that these figures are due to the presence of isolated nodes, that
is, nodes that did not participate in any communication in the observed month: if the
graph is reduced to non isolated nodes, essentially all values collapse.
 
\section{The sample is biased, and anyway it just represents $10\%$ of
humanity!}

As a first consideration, we invite the reader to observe that there is no such
things as a ``uniform'' or ``unbiased'' sample of a graph. One can, of course,
sample the \emph{nodes} or the \emph{arcs} of a graph, and consider the induced
subgraph, but there is no guarantee that the induced subgraph preserves the properties of interest of the whole
graph---much more sophisticated strategies are necessary, and in any case, 
it must be proved beforehand that the selected strategy creates an induced
subgraph that is sufficiently similar to the whole graph (whatever notion of
``similar'' we want to take into account).

In any case, let us take a step back and look for a moment at the
conditions of Milgram's experiment:
\begin{itemize}
\item \emph{number of pairs examined}: 296;
\item \emph{sample of the population}: 100 United States citizens living in
Boston, 96 random United States citizens living in Nebraska, 100 stockholders
living in Nebraska;
\item \emph{completed chains}: $\approx 22\%$;
\item \emph{definition of link}: instructions to send the letter only to a
``first-name acquaintance''.
\end{itemize}

Our case:
\begin{itemize}
\item \emph{number of pairs examined}: 250 millions of billions;
\item \emph{sample of the population}: 721 million people spread in several
continents;
\item \emph{completed chains}: $\approx 99.8\%$;
\item \emph{definition of link}: sharing a friendship link on Facebook.
\end{itemize}

We realize, obviously, that Facebook is not a random sample, and that being on
Facebook implies already sharing a mindset, or certain areas of interest.
We are also aware of the digital divide problem (that introduces a strong
geopolitical and economical bias) and that there are links on Facebook between
people that never met each other in person (e.g., gamers).

On the other hand, a random sample of 96 people from Nebraska is not a random
sample of the world population, either. And, again, we will never know if some
letters in the experiment actually passed through, say, two pen pals who never met in person.
What a lot of people did not realize is that, essentially, the only thing we
know about how people were involved in Milgram's experiment is that the sender
judged that it had a ``first-name acquaintance'' with the
receiver. The link between sender and receiver might have been in some
cases even \emph{weaker} than sharing a friendship link of Facebook.

There is, moreover, another important factor to take into account: since there
will be many first-name acquaintances who are \emph{not} on Facebook (and
hence not Facebook friends) some short paths will be missing. These two
phenomena will likely, at least in part, balance each other; so, although we do not have
(and cannot obtain) a precise proof of this fact, we do not think we are
losing or gaining much in considering the notion of Facebook friend as a surrogate of
first-name friendship. 

All in all, we see a definite progress in stating that the world is small.
Thanks to Facebook, which is the largest ever-created database of human
relationships, we have been able to make Milgram's experiment  (or at least the
part of it that has to do with measuring shortest paths) much more concrete and
objectively measurable. 

Nonetheless, let us take another step back and consider, for a moment, the
genius of a man who approached a mind-boggling (even for us, now)
problem  on a worldwide scale armed with three hundred postcards and an
incredibly clever experiment. Obtaining a result almost unbelievably close to
what we obtained using a number of pairs that is \emph{fifteen orders of
magnitude larger}. One is tempted to draw a comparison with Galileo's
celebrated mental experiment in the \textit{Dialogo sopra i due massimi sistemi
del mondo}~\cite{GalDSDMSM}: you do not need an expensive lab to test the principle of
relativity---you just need a ship, some butterflies and some fish. Of course,
once you do it, an expensive lab to check it thoroughly is definitely not a bad
idea.

\section{You measured the average distance, but degrees of
separation are algorithmic}

Just after we disseminated our paper, we learned that an experiment was trying
to settle the ``degree of separation'' problem, which was ``still unresolved''
using Facebook.\footnote{\texttt{http://smallworld.sandbox.yahoo.com/}.} We
were, of course, quite surprised. While we certainly did not ``resolve''
anything, it was difficult to imagine an experiment at present time with a
larger sample or significantly more precise measurements.

The point is the distinction between ``routing'' and ``distance''.
Milgram's postcard were routed locally (each sender did not know whether the
recipient was the best choice to get to the destination, i.e., if it lay on a
shortest path to the destination).
Apparently, the question is still unresolved because by studying Facebook we have only computed
the ``topological'', not the ``algorithmic'' degrees of separation.

We believe, however, that this is a red herring. Reading
carefully Travers and Milgram's papers~\cite{MilSWP,TMESSWP}, it is clear
that the very purpose of the authors was to estimate the number of
intermediaries:
the postcards were just a tool, and the details of the paths they followed were
studied only as an artifact of the measurement process. In the words of Milgram, 
the problem was defined by ``given two individuals selected randomly from the
population, what is the probability that the minimum number of intermediaries
required to link them is 0, 1, 2, \dots, $k$?''. Said otherwise, Milgram was
interested in estimating the \emph{distance distribution} of the acquaintance
graph.

The interest in efficient routing lies more in the
eye of the beholder (e.g., the computer scientist) than in Milgram's: if he had
at his disposal an actual large database of friendship links and algorithms like
the ones we used, he would have dispensed with the postcards altogether. Thus,
the fact that we measured
\emph{actual} shortest paths between individuals, instead of the paths of a
greedy routing, is a definite progress.
Routing is an interesting computer-science (and sociological) problem, but it
had little or no interest for Milgram---actually, the main interest in the routing process was
understanding the convergence of paths. From the paper: 
\begin{quote}
The theoretical machinery needed to deal with social
networks is still in its infancy. The empirical technique of this research has
two major contribution to make to the development of that theory. First it
sets an upper bound on the minimum number of intermediaries required
to link widely separated Americans. Since subjects cannot always foresee the
most efficient path to a targer, our trace procedure must inevitably produce
chains longer than those generated by an accurate theoretical model
which takes full account of all paths emanating from an individual.
\end{quote}

That said, the results obtained in Milgram's experiment are even more stunning
because the average routing distance they computed (with the provisos about
uncompleted chains discussed above) is so close to the average shortest-path
length. The latter observation seems to suggest that human beings are extremely
good at routing, so good that they almost route messages along the shortest possible
path. However, taking uncompleted paths into consideration gives a slightly
different twist to this remark: it seems that when someone felt confident
enough to continue the experiment, (s)he did so almost in the best possible way;
but more often than not, the experiment was stopped probably because
the message arrived at an individual that did not know how to route it further
efficiently.

Apart for the attempts to measure the routing distance in real-world social
graphs, there is an ever increasing focus on developing a theory of distributed
efficient routing on small worlds, starting from Kleinberg's intriguing notion
of navigability~\cite{KleNSW,KleSWP}; this is however outside of the scope of
our paper.

\section{Just add a few links here and there and we'll all be at one degree of
separation}

Another, closely related, question is: ``We have seen that the degree of
separation has constantly decreased since 2008, reaching its current value. What can we expect for the future?''

To answer the above comment/question, 
notice that the average distance is
\[
\sum_{k>0}kP_k /r,
\]
where $P_k$ is the number of pairs at distance exactly $k$ and $r$ is the
number of reachable pairs, which is $n^2$ if and only if the graph is strongly
connected. Of course, if we have bounds $B_k\geq P_k$ for some $1\leq
k\leq\ell$, it is immediate to see that, if $\sum_{k=1}^{\ell -1 }B_k\leq r$
then
\begin{equation}
\label{eqn:bound}
\sum_{k>0}kP_k \geq \sum_{k=1}^{\ell-1} kB_k + \ell\Bigl(r -
\sum_{k>0}B_k\Bigr). 
\end{equation}
Now, depending on how much you want to consider a graph similar to the Facebook
graph described in~\cite{BBRFDS}, there are many ways to generate some $B_k$'s.

\paragraph{First bound (depending on $n$, $m$ and $D$).}
There are instrinsic bounds on the number
of short paths you can generate when the number of neighbours of a node is limited.
The simplest observation is that (letting $D$ be the maximum degree and $m$ be
the number of arcs in the graph, i.e., twice the number of edges) you cannot
have more than $m$ pairs at distance one, $mD$ pairs at distance 2,
and so on; more precisely, we can set $B_k=mD^{k-1}$,
getting (from (\ref{eqn:bound})) the lower bound
\[
\sum_{k>0}kP_k\geq m + 2mD + 3(r - m - mD)
\]
provided that $m + mD\leq r$; in the case of Facebook ($D=5000$,
$n\approx 721\times 10^6$, $r=5\times 10^{17}$, $m\approx 69\times 10^9$) the
inequality $m+mD\leq r$ is satisfied and the lower bound obtained is $\approx
2.999$.
In other words, no graphs with the same number of nodes, arcs and maximum outdegree of the graph we considered can have an average distance smaller than
$2.999$.

% \paragraph{Second bound (depending on $n$, $m$).}
% There are of course obvious improvements. Let
% $d(x)$ be the degree of node $x$. The number of paths of length two is
% given by the sum $s$ of entries of the square of the adjacency matrix $M$ of the
% graph. But
% \begin{multline*}
% s =\sum_i\sum_j\sum_k m_{ik}m_{kj} =\sum_i\sum_k m_{ik}\sum_j m_{kj}
% \\=\sum_i\sum_k m_{ik}d(k) =\sum_kd(k)\sum_i m_{ik} =\sum_kd(k)^2. 
% \end{multline*}
% By de Caen's bound~\cite{deCUBSSDG}, the sum of the squares of the degrees is at
% most $C(n,m)=(m/2)(m/n+n-1)$. We can thus improve our bound as follows:
% \[
% \sum_{k>0}kP_k\geq m + 2C(n,m) + 3(n^2 - m -
% C(n,m))
% \]
% provided that $m + C(n,m)\leq n^2$, and we get $\approx 2.99998$. Note that this
% bound \emph{does not depend on the actual degree distribution}, but rather just
% on the number of nodes and arcs, or, equivalently, on the number of nodes and on
% the average degree $d=m/n$:
% \begin{multline*}
% m + 2C(n,m) + 3(n^2 - m - C(n,m))/n^2 \\
% = d/n + d(d/n+1-1/n) +\\
%  3( 1 - d/n - (d/2)(d/n+1-1/n) ).
% \end{multline*}
% In other words, as humanity gets bigger, assuming some intrinsic limit on
% average on the number of acquaintainces we are able to memorise or handle
% (``Dunbar's number''\cite{}) the bound just gets closer to 3. Actually, we just
% need $d=o\bigl(\sqrt n\bigr)$, which is a very mild assumption.

\paragraph{Second bound (depending on the degree sequence).}
To improve over the previous trivial bound, we can use the actual degree
distribution.\footnote{The degree distribution is publicly available as part of
the dataset associated with~\cite{BBRFDS}.} This is a bit like answering to the
question: what if some omniscent being ``rewired'' Facebook in an optimised way
to reduce the average distance as much as possible, but leaving each user with
its current number of friends? Let us first notice that $P_2$ can be bounded by
$\sum_x d(x)^2$, which, being the sum of entries of the square of the adjacency
matrix, is an upper bound for the number of pairs at distance 2. Providing a
good bound for $P_3$ is slightly more difficult:

\begin{theorem}
\label{teo:P3}
Let $d_0\geq d_1\geq \dots d_{n-1}$ be the degree sequence of the graph, $s =
\sum_{i=0}^{n-1} d_i^2$ and define, for every $t$,
\[
	\delta(t)=\sum_{i=0}^{d_t-1} d_i.
\]
Then $P_3$ (the number of pairs of nodes at distance exactly 3) can be bounded
by
\[
	P_3 \leq \sum_{k=0}^{\ell} d_k \delta(k) + d_{\ell+1} \Bigl( s -
	\sum_{k=0}^{\ell}\delta(k) \Bigr)
\]
where $\ell$ is the greatest integer such that
$\sum_{k=0}^{\ell}\delta(k)< s$.
\end{theorem}

\begin{proof}
We can bound $P_3$ from above by counting the number $p$ of tuples
$(u_i,v_i,w_i,z_i)$ corresponding to paths of length 3. Let
$V=\{v_0,\dots,v_{k-1}\}$ be the set of nodes appearing as second component in
at least one such tuple, sorted by non-increasing node degree; clearly $p\leq
d(v_0)\pi(v_0)+\dots+d(v_{k-1})\pi(v_{k-1})$ where $d(x)$ is as usual the degree
of $x$ and $\pi(x)$ is the number of paths of length 2 starting from $x$: this
is because every single path of length 3 of the form $(-,v_i,-,-)$ is obtained
by choosing a neighbor of $v_i$ and a path of length 2 leaving from $v_i$.

Observe that $\pi(v_0)+\dots+\pi(v_{k-1})$ cannot be larger than $s$
(because the latter is an upper bound to the number of paths of length 2 in the graph).
Now, of course, for every $t=0,\dots,k-1$, $d(v_t)\leq d_t$, so $p\leq
d_0\pi(v_0)+\dots+d_{k-1}\pi(v_{k-1})$; it is convenient to think of the latter
as a summation of a list $L$ of length
$s\geq\pi(v_0)+\dots+\pi(v_{k-1})$, where 
$d_0$ occurs $\pi(v_0)$ times, $d_1$ occurs $\pi(v_1)$ times etc., and at the
end of the list $0$ occurs enough times to reach the desired length.

Now $\pi(v_t)$ can
be bounded from above by the number of paths of length 2 leaving from a node of
degree $d_t$. But the latter can be obtained by choosing at the first step the
$d_t$ nodes with largest degree, and summing up their degree; that is,
$\pi(v_t)\leq \delta(t)$. So we can safely substitute the above list $L$ with
another list $L'$ of the same length where $d_0$ is repeated $\delta(0)\geq
\pi(v_0)$ times, $d_1$ is repeated $\delta(1)\geq \pi(v_1)$ times etc. The
resulting list $L'$ dominates $L$ elementwise, hence the thesis.
\end{proof}

Plugging $B_1=m$, $B_2=\sum_{i=0}^{n-1} d_i^2$ and $B_3$ as in
Theorem~\ref{teo:P3}, and using the actual degree sequence of Facebook, we
obtain $\approx 3.6$. Thus, Facebook is essentially just one step (distance or degree
doesn't matter) away from the best possible, given that every individual keeps the
current number of friends.

\section{It's just because of the nodes with very high degree that we observe
such a low value}

Since the first studies on the structure of complex graphs~\cite{BAJPDWWW}, and
in particular of social networks, the degree distributions have been a central topic on which
many authors focused, concluding that both in- and out-degrees exhibit a
heavy-tailed distribution: this fact implies that there are many nodes whose
degree largely exceeds the average. It is a widely assumed tenet that those
nodes, sometimes referred to as \emph{hubs}, represent a sort of ``social glue''
that keeps the whole network structure together and that shortcut
friendship paths.
In the case of social networks, such as Twitter or Facebook, hubs are superstars
like Lady Gaga or Barack Obama, whose account often do not even correspond to
real persons. 

But, is this the case? In our analysis of the Facebook graph we excluded
\emph{pages} (the accounts that people may ``like''), and standard accounts
have a hardwired limit of $5\,000$ friends. Nonetheless, we cannot rule out the
possibility that there are some fake celebrity accounts remaining in the graph we studied.

The general question we are asking can be restated as follows: take a social
network and start removing the nodes of largest degrees; how much does the distribution of
distances change? in particular: how does the average distance change
(presumably: increase)? 
We considered this question in a previous paper~\cite{BRVRSN}
(see also~\cite{BRVRSWSG}), where we actually studied the more general problem of
which removal strategies are more disruptive under the viewpoint of distance distributions.

We report an anticipation of a subset of the results of~\cite{BRVRSWSG}, as they
suggest that high-degree node removal is not going to cause drastic changes in
the structure of the network. We show results for a small\footnote{Similar
results have been obtained with a lesser degree of precision on a snapshot of a
100 million pages in~\cite{BRVRSN}; computations are underway to obtain
high-precision data similar to what we report here about the smaller snapshot,
and the results will be included in the final version of this paper.} snapshot of the Indian web (\texttt{.in}), for the Hollywood
co-starship graph, for a snapshot of the LiveJournal network kindly provided by the authors of~\cite{CKLCSN}, and a snapshot of the Orkut network kindly provided by the authors of~\cite{MMKMAOSN}.\footnote{All these datasets are public and available at \texttt{http://law.dsi.unimi.it/}. The identifiers of the datasets are \texttt{in-2004}, \texttt{hollywood-2011}, \texttt{ljournal-2008} and
\texttt{orkut-2007}.}

The results we obtained are the following. Removing largest-degree nodes does
affect the average distance on web graphs: after the removal of $30\%$ of the
arcs\footnote{We emphasize that we remove nodes (in decreasing order of their
in-degree) and all incident edges, but count how many \emph{arcs} are removed, because it
is the number of deleted arcs that determines the expected loss in
connectivity. We invite the reader to consult~\cite{BRVRSN} for more details.}
the average distance gets increased of about $24\%$.
Nonetheless, the same removal strategy seems to have a weaker impact on genuine social networks: under
the same condition, the increase in average distance ranges between $8\%$ and
$11\%$ (see Table~\ref{tab:avgdi}).

\begin{table}
\caption{\label{tab:avgdi}Change in average distance of web and social graphs
after removing the largest (in-)degree nodes. The removal process is stopped
when the number of arcs removed reaches the $10\%$ and $30\%$.}
\centering
\begin{tabular}{|l|D{.}{.}{2}|D{.}{.}{2}@{\hspace{2mm}}r|D{.}{.}{2}@{\hspace{2mm}}r|}
\hline Graph 	& \multicolumn{1}{|c|}{original} & \multicolumn{2}{|c|}{$10\%$} &
	\multicolumn{2}{|c|}{$30\%$}
	\\ \hline
	\texttt{.in} 
		& 15.34 & 16.11 & $(+5.0\%)$ & 18.98 & $(+23.7\%)$ \\
	\hline
	Hollywood
		& 3.92 & 4.02 & $(+2.5\%)$ & 4.23 & $(+7.9\%)$ \\
	LiveJournal
		& 5.99 & 6.15 & $(+2.7\%)$ & 6.55 & $(+9.3\%)$ \\
	Orkut
		& 4.21 & 4.43 & $(+5.2\%)$ & 4.67 & $(+10.9\%)$ \\
	\hline
\end{tabular}
\end{table}

Nonetheless, we are actually missing a very important point: in the social
networks we studied, removing 30\% of the arcs actually does not change the
percentage of reachable pairs, whereas in web graphs the percentage (which is
already lower) is reduced by a half. As we discussed in
Section~\ref{sec:dist}, the average distance turns out again to be a very rough
and unrealiable measure when the number of unreachable pairs is large.

Thus, in Table~\ref{tab:avgdh} we show what happens to the harmonic diameter.
The results show that the increase for social networks is very modest (less than
$20\%$ after the removal of as many as the $30\%$ of the arcs), whereas for web
graphs the harmonic diameter almost triplicates! This confirms again that the
harmonic diameter is more reliable value to be associated to the ``tightness''
or ``connectedness'' of a network.

\begin{table}
\footnotesize
\caption{\label{tab:avgdh}Change in harmonic diameter of web and social graphs
after removing the largest (in-)degree nodes. The removal process is stopped
when the number of arcs removed reaches the $10\%$ and $30\%$.}
\centering
\begin{tabular}{|l|D{.}{.}{2}|D{.}{.}{2}@{\hspace{2mm}}r|D{.}{.}{2}@{\hspace{2mm}}r|}
\hline Graph 	& \multicolumn{1}{|c|}{original} & \multicolumn{2}{|c|}{$10\%$} &
	\multicolumn{2}{|c|}{$30\%$}
	\\ \hline
	\texttt{.in} 
		& 32.26 & 47.03 & $(+45.8\%)$ & 87.68 & $(+171.8\%)$ \\
	\hline
	Hollywood
		& 4.08 & 4.12 & $(+1.0\%)$ & 4.40 & $(+7.8\%)$ \\
	LiveJournal
		& 7.36 & 7.74 & $(+5.2\%)$ & 8.67 & $(+17.8\%)$ \\
	Orkut
		& 4.06 & 4.33 & $(+6.7\%)$ & 4.61 & $(+13.6\%)$ \\
	\hline
\end{tabular}
\end{table}

We remark that LiveJournal and Orkut are people-to-people friendship networks as
Facebook (note, however, that LiveJournal is directed).
We believe that the resistance to high-degree removal is actually a common
phenomenon in such networks, which prompts us to conjecture that similar
node-removal prodedures will not change Facebook average distance or harmonic
diameter significantly, albeit we have no empirical data to support our
hypothesis at this point.

Actually, a more general conclusion obtained in the cited paper~\cite{BRVRSN}
is that social networks seem very robust to node removal, and we could not find
any node order that determined radical changes in the distance distribution. 
This observation leaves an intriguing question still open to debate: if hubs are
not the inherent cause behind short distances, then what is the \emph{real}
reason of this phenomenon?

\section{Are you saying that Facebook reduced the average distance between
people?}

Some of the comments in the general press took the outcomes of our experiments
as an evidence that online social networks (such as Facebook) reduced the
average distance between people; of course, this was not the purpose (neither
the content) of the experiment and in any case there is no direct way to know if
this is true or not, because our measurements \emph{are performed on Facebook}. We can see,
however, that the distance between Facebook users constantly decreased over
time:
it used to be $5.28$ in 2008, $5.06$ in 2010 and $4.74$ in our most recent dataset. Whether this decrease is \emph{due} to Facebook, or whether it
simply Facebook reflecting better and better the situation in the ``real world''
is hard to say. In the former case, as someone suggested, we would be observing
a reduction in path lengths due probably to the presence of \emph{weak
ties}~\cite{GraSWT} that hardly correspond to a real friendship relation and
would probably not even show up in a non--electronically-mediated environment.

Understanding how online social networks are changing our way of interacting,
communicating and thinking is absolutely beyond the scope of our paper, whose
aim was much humbler and certainly not as far-reaching. We believe,
however, that giving a concrete and realistic explanation of what is going on
requires a co-ordinated effort and calls for an interdisciplinary endeavor,
putting together sociology, psychology, computer science and mathematics. This
is, we think, one of the most important challenges for people working in these
disciplines, with yet unknown consequences of philosophical, social and even
economical value.

\bibliography{biblio,fb,fdsr,law}

\end{document}

%% file: scripts/harmonic.table
		 2007  &  15083.99 ($\pm$298.82) &  51.07 ($\pm$1.50) &  3760.77 ($\pm$161.28) &  4.16 ($\pm$0.14) &  6.33 ($\pm$0.26) \\
		 2008  &  23.66 ($\pm$0.75) &  4.37 ($\pm$0.15) &  6.44 ($\pm$0.21) &  4.61 ($\pm$0.16) &  5.74 ($\pm$0.24) \\
		 2009  &  4.74 ($\pm$0.11) &  4.37 ($\pm$0.11) &  4.71 ($\pm$0.11) &  4.67 ($\pm$0.16) &  5.07 ($\pm$0.21) \\
		 2010  &  3.92 ($\pm$0.13) &  3.90 ($\pm$0.16) &  4.24 ($\pm$0.18) &  4.68 ($\pm$0.15) &  5.03 ($\pm$0.21) \\
		 2011  &  3.76 ($\pm$0.11) &  3.93 ($\pm$0.16) &  4.29 ($\pm$0.18) &  4.23 ($\pm$0.13) &  4.70 ($\pm$0.30) \\
		 current  &  3.68 ($\pm$0.10) &  3.69 ($\pm$0.20) &  3.90 ($\pm$0.13) &  4.45 ($\pm$0.11) &  4.59 ($\pm$0.13) \\